\documentclass[12pt]{article}
\usepackage{amsmath,amsthm,amssymb}
\usepackage{array, longtable}
\newtheorem{Theorem}{Theorem}[section]
\newtheorem{lem}[Theorem]{Lemma}
\newtheorem{Remark}[Theorem]{Remark}
\newtheorem{Definition}[Theorem]{Definition}
\newtheorem{Corollary}[Theorem]{Corollary}
\newtheorem{Proposition}[Theorem]{Proposition}
\newtheorem{Example}[Theorem]{Example}

\numberwithin{equation}{section}
\numberwithin{table}{section}

\begin{document}

\title{Generalized Pair Weights of Linear Codes and Linear Isomorphisms Preserving Pair Weights\footnote{
 E-Mail addresses: hwliu@mail.ccnu.edu.cn (H. Liu), panxu@mails.ccnu.edu.cn (X. Pan)}}

\author{Hongwei Liu,~Xu Pan}
\date{\small
School of Mathematics and Statistics, Central China Normal University\\Wuhan, Hubei, 430079, China\\
}
\maketitle

\begin{abstract}
In this paper, we first introduce the notion of generalized pair weights of an $[n, k]$-linear code over the finite field  $\mathbb{F}_q$ and the notion of pair $r$-equiweight codes, where $1\le r\le k-1$. Some basic properties of generalized pair weights of linear codes over finite fields are derived. Then we obtain a necessary and sufficient condition for an $[n,k]$-linear code to be a pair equiweight code, and we characterize pair $r$-equiweight codes for any $1\le r\le k-1$. Finally, a necessary and sufficient condition for a linear isomorphism preserving  pair weights between two linear codes is obtained.


\end{abstract}
\section{Introduction}
In 1950, Hamming introduced the notions of Hamming weight (usually written $w_H$) and Hamming distance (usually written $d_H$)
which would serve  as the basis for modern coding theory.
The notion of
generalized Hamming weights appeared  in the
1970's and has become an important research
object in coding theory after Wei's work \cite{W} in 1991.
More specifically,
let $\mathbb{F}_q$ be the finite field with  $q$ elements, where $q=p^e$ and $p$ is a prime.
An $[n,k]$-linear code $C$ of length $n$ over $\mathbb{F}_q$ is an $ \mathbb{F}_q$-subspace of dimension $k$ of $ \mathbb{F}_q^{n}$.
Let $r$ be an integer with $1\leq r\leq k$ and let $V$ be a subspace of dimension $r$ of $C$.
The Hamming support of $V$ is defined by $\chi_H(V)=\{i\,|\,0\leq i\leq n-1, \hbox{$\exists (c_0,\cdots,c_{n-1})\in V$ such that $c_i\neq0$}\}.$
Consequently, the {\em $r${\rm th} generalized Hamming weight}  of a
linear code $C$ over $\mathbb{F}_q$ is defined by
$d_H^r(C)=\min \{|\chi_H(V)|\,|\,V \,\hbox{is an $r$-dimensional subpace of $C$\}}$.
It is obvious that $d_H^1(C)$ is just the minimum Hamming
distance $d_H(C)$ and the set  $\{d_H^1(C), d_H^2(C), \cdots, d_H^k(C)\}$ is called the {\em generalized Hamming weight hierarchy}
of $C$.

Wei  \cite{W} showed  that the generalized Hamming weight hierarchy of a code is of great importance in the sense that it
features the performance of a linear
code completely
and has a close connection with cryptography;
a series of good properties  on the generalized Hamming  weight hierarchy of a code were also exhibited in \cite{W}.
Since then, lots of works have been done in computing and describing   the  generalized Hamming weight hierarchies  of certain codes,
see, for example,
 \cite{B},  \cite{JFW}, \cite{TV} and
 \cite{YF}.

The  MacWilliams extension theorem plays a central role in coding theory.
MacWilliams \cite{M} and later Bogart, Goldberg, and Gordon
\cite{BGG} proved that, every linear isomorphism preserving the Hamming weight between two linear codes over finite fields can be extended to a monomial transformation. This classical result was known as  MacWilliams extension theorem.
In \cite{J}, Wood proved MacWilliams extension theorem for all linear codes over finite Frobenius rings equipped with the Hamming weight. In the commutative case, the author showed that the Frobenius property was not only sufficient but also
necessary. In the non-commutative case, the necessity of the Frobenius property was proved in \cite{J1}.

With the development of information theory, a number of new metrics have been introduced to coding theory, for example, the Lee metric, the burst metric, homogeneous metric£¬ etc. In 2011, motivated by the limitations of the reading process in high density
data storage systems,  Cassuto and
Blaum \cite{CB} introduced a new metric framework, named symbol-pair distance,
to protect against pair errors in
symbol-pair read channels, where the outputs are overlapping
pairs of symbols.
The seminal work \cite{CB} has established
relationships between the minimum Hamming distance of
an error-correcting code and the minimum pair distance,
has found methods for code constructions and decoding,
and has obtained lower and upper bounds on the code sizes.
In \cite{C}, the authors established a Singleton-type bound for symbol-pair codes and constructed MDS symbol-pair codes (meeting this Singleton-type bound), which is called the {\em maximum pair distance separable} (MPDS) code in this paper.
Several works have been done on the constructions of MPDS codes, see, for example,
\cite{KZL},  \cite{LG}, \cite{DGZZ} and  \cite{CLL}.
In \cite{LXY},  Liu, Xing and Yuan presented the list decodability of symbol-pair codes and a list decoding algorithm of Reed-Solomon codes beyond the Johnson-type bound in the pair weight. In \cite{DNSS} and \cite{DWLS}, the authors calculated the symbol-pair distances of repeated-root constacyclic codes of lengths $p^s$ and $2p^s$, respectively. Yaakobi, Bruck and Siegel \cite{YBS} generalized the notion of symbol-pair weight to $b$-symbol weight.  Yang, Li and Feng \cite{YL} showed the Plotkin-like bound for the $b$-symbol weight and presented a construction on irreducible cyclic codes and constacyclic codes meeting the Plotkin-like bound.

As mentioned above,  symbol-pair distance is a new metric model compared to the classical Hamming distance.
Therefore, it is natural to ask how theorems surrounding classical coding theory generalize  to the current symbol-pair framework. This generalization would have some potential applications in cryptography. Indeed, as indicated in the proceeding paragraph, several bounds on the minimum symbol-pair distance have been established, including the Singleton-type bound, the Johnson-type bound and the   Plotkin-like bound.

In this paper, we introduce the notion of generalized pair weights of linear codes over finite fields, basic properties of generalized pair weights are derived. In particular, the Singleton Bound respect to generalized pair weights are established, and a necessary and sufficient condition for a linear code to be an MPDS code is obtained.  For an $[n,k]$-linear code, we introduce the notion of the pair equiweight code and the pair $r$-equiweight code for any $1\le r\le k-1$. A necessary and sufficient condition for a linear code to be a pair equiweight code is derived. Moreover, we characterize pair $r$-equiweight codes. Note that MacWilliams extension theorem tells that every linear isomorphism preserving the Hamming weight bewteen two lienar codes can be induced by a monomial matrix. Unfortunately, a linear isomorphism induced by a permutation matrix may not preserve the pair weight between two linear codes. In this paper, we provide a necessary and sufficient condition for a linear isomorphism preserving pair weights between two linear codes.

This paper is organized as follows. Section 2 provides some preliminaries, and we introduce generalized pair weights of linear codes,  and give a characterization of the pair weight of arbitrary codeword of a linear code. In Section~3, basic properties of generalized pair weights of linear codes are provided, and some other results are also given. In Section~4, we give a necessary and sufficient condition for a linear code to be a pair equiweight code. We obtain a necessary condition and a sufficient condition for an $[n,k]$-linear code to be a pair $r$-equiweight code. Section~5 studies linear isomorphisms preserving pair weights of linear codes, we obtain a necessary and sufficient condition for a linear isomorphism preserving pair weights. In particular, we provide an algorithm to determine whether a linear code is a pair equiweight code, and whether an isomorphism between two linear codes preserves pair weights. We explain why this algorithm is more efficiently.

\section{Preliminaries}

Throughout this paper, let $\mathbb{F}_q$ be the finite field of order $q$, where $q=p^e$ and $p$ is a prime number. Let $n$ be a positive integer, and let $\mathbb{F}_q^{n}$ be the $n$-dimensional vector space over $\mathbb{F}_q$. An $\mathbb{F}_q$-subspace $C$ of dimension $k$ of $\mathbb{F}_q^{n}$ is called an $[n,k]$-linear code. The dual code $C^{\perp}$ of $C$ is defined as
$$
C^{\perp}=\{{\bf x}\in \mathbb{F}_q^n \,| \,  {\bf c}\cdot{\bf x} =0, \forall \, {\bf c}\in C\},
$$
where $ ``-\cdot-" $ is the standard Euclidean inner product.  We assume all codes in this paper are nonzero linear codes.

\begin{Definition}\label{AA}(\cite{CB})
For any $\mathbf{x},\mathbf{y} \in \mathbb{F}_{q}^{n}$, the pair distance between $\mathbf{x}$ and $\mathbf{y}$ is defined as
$$d_{p}(\mathbf{x},\mathbf{y})=\big{|}\{0\leq i \leq n-1|(x_{i},x_{i+1})\neq(y_{i},y_{i+1})\}\big{|},$$ where the indices are taken modulo $n$. The pair weight of $\mathbf{x}$ is defined as $w_{p}(\mathbf{x})=d_{p}(\mathbf{x},{\bf 0}).$
\end{Definition}

The {\it minimal pair distance} of a code $C$ over $\mathbb{F}_q$ is defined as $$d_{p}(C)=\min_{\mathbf{c}\ne\mathbf{c'} \in C}\,d_{p}(\mathbf{c},\mathbf{c}').$$
The {\it minimal pair weight} of $C$ is defined as $\min\{w_p({\bf c})\,|\,{\bf 0}\ne {\bf c}\in C\}.$ Note that if $C$ is an $[n,k]$-linear code, then $d_{p}(C)=\min\{w_p({\bf c})\,|\,{\bf 0}\ne {\bf c}\in C\}$.

An $[n,k]$-linear code $C$ over $\mathbb{F}_q$ is called a {\it pair equiweight code} if any nonzero codeword of $C$ has the same pair weight.

The generalized Hamming weights of any $ \mathbb{F}_q$-subspace of $ \mathbb{F}_q^{n}$ and the $r$-minimal Hamming weight of an $[n,k]$-linear code $C$ over $\mathbb{F}_q$ for $1\leq r \leq k$ were defined by Wei \cite{W}.

\begin{Definition}\label{BB}(\cite{W})
Let $D$ be an $ \mathbb{F}_q$-subspace of $ \mathbb{F}_q^{n}$. The Hamming support of $D$, denoted by $\chi_H(D)$, is the set of all non-always-zero bit positions of $D$, i.e., $$\chi_H(D)=\{0\leq i \leq n-1\,|\,\exists \,\mathbf{x}=(x_{0}, x_1, \cdots,x_{n-1})\in D, x_{i}\neq0\},$$ and the generalized Hamming weight of $D$ is defined as $w_{H}(D)=|\chi_H(D)|$.
\end{Definition}

It is quite natural that we can assume $\chi_H(D)\subseteq \mathbb{Z}/n\mathbb{Z}$, the ring of integers modulo $n$.

\begin{Definition}\label{KK}(\cite{W})
Let $C$ be an $[n,k]$-linear code over $\mathbb{F}_q$. For $1\leq r \leq k$, the $r$-minimal Hamming weight of $C$ is defined as $d_{H}^{\,r}(C)=\min\{w_{H}(D)\,|\,D\leq C, \dim(D)=r\}$.
\end{Definition}

Note that if $r=1$, the $1$-minimal Hamming weight of $C$ is just the minimal Hamming weight of $C$.  In \cite{W}, the following result was proved.

\begin{lem} \label{monotonicity}(\cite[Theorem 1]{W})
Let $C$ be an $[n,k]$-linear code over $\mathbb{F}_q$. Then we have $$1\leq d_{H}^{\,1}(C)<d_{H}^{\,2}(C)<\cdots<d_{H}^{\,k-1}(C) < d_{H}^{\,k}(C)\leq n.$$
\end{lem}

The set $\{d_H^{\,1}(C), d_H^{\,2}(C),\cdots, d_H^{\,k}(C)\}$ is called the {\it generalized Hamming weight hierarchy} of $C$.
%Then we give core definitions of this paper which is the generalized pair weight and $r$-minimal pair weight.

 In 2003, Fan and Liu  \cite{FL}  introduced the Hamming $r$-equiweight code for an $[n,k]$ linear code over $\mathbb{F}_q$, where  $1\leq r\leq k-1$.

\begin{Definition}
Let $C$ be an $[n,k]$-linear code over $\mathbb{F}_q$ and $1\leq r\leq k-1$. The code  $C$ is called a Hamming $r$-equiweight code if $d_{H}^{\,r}(C)=w_{H}(D)$ for any subspace $D$ of dimension $r$ of $C$.
\end{Definition}
Note that if $r=1$, the Hamming $r$-equiweight code is just the Hamming equiweight code as usual. The properties of this class of codes are also obtained in \cite{FL}.
%Let $M_n(\mathbb{F}_q)$ be the set of all $n\times n$ matrices over $\mathbb{F}_q$. For $A\in M_n(\mathbb{F}_q)$, let $A^T$ denote the transpose of $A$. Let ${\rm GL}_{n}(\mathbb{F}_q)$ be the set of all $n\times n$ invertible matrixes over $\mathbb{F}_q$.
%A {\it monomial matrix} over $\mathbb{F}_q$ is a square matrix such that in every row and in every column there is exactly one nonzero element. Let ${\rm MO}_{n}(\mathbb{F}_q)$ denote the set of all the $n\times n$ monomial matrices over $\mathbb{F}_q$.

%The following theorem is called classical MacWilliams extension theorem, which was first proved by (MacWilliams \cite{M}).  Bogart, Goldberg, and Gordon provided an alternative proof in 1978 (\cite{BGG}).

%\begin{Proposition}[MacWilliams Extension Theorem (\cite{M},\cite{BGG})]\label{Mac}
%Let $C$ and $\tilde{C}$ be two $[n,k]$-linear codes over $\mathbb{F}_q$. Then there exists an $\mathbb{F}_{q}$-linear isomorphism $f:C\rightarrow \tilde{C}$ which preserves Hamming weights if and only if there exists a monomial matrix $M \in MO_{n}(\mathbb{F}_q)$ such that $f(\mathbf{c})=\mathbf{c}M$ for all $\mathbf{c}\in C$.
%\end{Proposition}

In this paper, we introduce the notion of generalized pair weights of any $ \mathbb{F}_q$-subspace of $ \mathbb{F}_q^{n}$ and $r$-minimal pair weight of $[n,k]$-linear codes over $\mathbb{F}_q$, where $1\leq r \leq k$. We will study their properties in this paper.

\begin{Definition}
Let $D$ be an $ \mathbb{F}_q$-subspace of $ \mathbb{F}_q^{n}$. The pair support of $D$ is defined as $$\chi_{p}(D)=\{0\leq i \leq n-1\,|\,\exists\, \mathbf{x}=(x_{0},\cdots,x_{n-1})\in D, (x_{i},x_{i+1})\neq(0,0)\},$$   where the indices are taken modulo $n$. The generalized pair weight of $D$ is defined as $w_{p}(D)=|\chi_{p}(D)|$.
\end{Definition}

\begin{Definition}
Let $C$ be an $[n,k]$-linear code over $\mathbb{F}_q$. For $1\leq r \leq k$, the $r$-minimal pair weight of $C$ is defined as $d_{p}^{\,r}(C)=\min\{w_{p}(D)\,|\,D\leq C, \dim(D)=r\}$. The set $\{d_p^{\,1}(C), d_p^{\,2}(C),\cdots, d_p^{\,k}(C)\}$ is called the  generalized pair weight hierarchy of $C$.

\end{Definition}
\begin{Remark}\label{2.8}
If $r=1$, the $1$-minimal pair weight  $d^{\,1}_p(C)$ of the code $C$ is just the minimal pair weight $d_p(C)$ of $C$. In \cite{C}, we know $d_{p}(C)\leq n-k+2$. If $C$ satisfies $d_{p}(C)= d_p^{\,1}(C)=n-k+2$, then we call $C$ a {\it maximum pair distance separable } (MPDS) code.
 \end{Remark}

%In \cite{F}, Forney introduced the definition of the LDP for the Hamming weight. Analogously, we give the definition of LDP for the pair weight which is essentially the same as the generalized pair weight hierarchy proved in Theorem~\ref{essent}.

Let $J$ be a subset of $\{0,1,\cdots,n-1\}$. The {\it subcode $C_J$ of  a code $C$ of length $n$ for pair weights} is defined to be: $$C_J=\{\mathbf{c}=(c_0,c_1,\cdots,c_{n-1})\in C \,|\,(c_i,c_{i+1})=(0,0)\,\,\,\forall \,i\notin J\}.$$

By the definition of $C_{J}$, we know that $C_J=C$ when $J=\{0,1,\cdots,n-1\}$ and $C_J=\mathbf{0}$ when $J=\emptyset$. Also we have $C_{J_1}\subseteq C_{J_2}$ if $J_1\subseteq J_2$.

\begin{Definition}
Let $C$ be an $[n,k]$-linear code over $\mathbb{F}_q$, let $J\subseteq \{0,1,\cdots,n-1\}$. Let $C_J$ be defined as above. For $1\le r\le k$, let $m_r(C)=\min\limits_J\{|J|\,|\,\dim(C_J)=r\}$. Then the following sequence is called  the {\it length/ dimension profile (LDP) for the pair weight} of $C$:
$$\mathbf{m}(C)=\{m_1(C),\,m_2(C)\,\cdots, m_k(C)\}.$$
\end{Definition}

%We are going to give a characterization of the pair weight of an element in a linear code $C$ used in Section 4 and Section 5. Before we do that, we need introduce some notations.

Let $U$ be an $\mathbb{F}_q$-vector space of dimension $k$. We denote by $\langle V, W\rangle$ the subspace generated by the subspaces $V, W$ of $U$, and let $U/W$ denote the quotient space modulo $W$. For any $ r, k\in \mathbb{N}$, let
\begin{align*}
  {\rm PG}^{r}(U)=\{V\leq U\,|\,\dim(V)=r\}&,  & {\rm PG}^{\leq r}(U)=\{V\leq U\,|\,\dim(V)\leq r\}.
\end{align*}
If $V=\{{\bf 0}\}$, then $\dim(\{\mathbf{0}\})=0$ and ${\rm PG}^{0}(U)=\{\{\mathbf{0}\}\}$.
Let $n_{r,k}$ denote the number of all $r$-dimensional subspaces of an $k$-dimensional vector space. When $r> k$, we let $n_{r,k}=0$. Then it is easy to see that
$$n_{r, k}= \left\{ \begin{array}{ll}
1,  & \textrm{if $r=0\ ;$}\\

\prod\limits_{i=0}^{r-1}\frac{q^{k}-q^{i}}{q^{r}-q^{i}},  & \textrm{if $1\leq r \leq k ;$}\\

0,  & \textrm{if $r> k .$}
\end{array} \right.
$$
Let $C$ be an $[n,k]$-linear code with a generator matrix $G=(G_{0},\cdots,G_{n-1})$,  where $G_i$ is the column vector of $G$. For any $V\in {\rm PG}^{\leq 2}(\mathbb{F}_q^{k})$, the function $m_{G}: {\rm PG}^{\leq 2}(\mathbb{F}_q^{k}) \to \mathbb{N}$ is defined as follows.
$$
m_{G}(V)=\big{|}\{0\leq i\leq n-1\,|\,\langle G_{i},G_{i+1}\rangle= V\}\big{|},
$$  where the indices are taken modulo $n$.
We define the function $\theta_{G}: {\rm PG}^{\leq k}(\mathbb{F}_q^{k}) \to \mathbb{N}$ to be $$\theta_{G}(U)=\sum_{V\in {\rm PG}^{\leq 2}(U)}m_{G}(V)$$ for any $U\in {\rm PG}^{\leq k}(\mathbb{F}_q^{k})$.

For an $[n,k]$-linear code $C$ over $\mathbb{F}_q$ with a generator matrix $G$, we know that for any $1\leq r \leq k $ and a subspace $D$ of dimension $r$ of $C$, there exists an unique subspace $\tilde{D}$ of dimension $r$ of $\mathbb{F}_q^k$ such that $D=\tilde{D}G=\{{\bf y}G\,|\,{\bf y}\in \tilde{D} \}$. In particular, for any nonzero codeword ${\bf c}\in C$, there exists an unique nonzero vector ${\bf y}\in \mathbb{F}_q^k$ such that ${\bf c}={\bf y}G=({\bf y}G_0,{\bf y}G_1,\cdots, {\bf y}G_{n-1})$, where $G=(G_0,\cdots, G_{n-1})$.

\begin{Proposition} \label{r pair weight}
Assume the notations are given above. Then $w_{p}(D)= n-\theta_{G}(\tilde{D}^{\bot})$ for any subspace $D$ of $C$, where $\tilde{D}$ is the unique corresponding subspace of $D$. In particular, $w_{p}(\mathbf{c})= n-\theta_{G}(\langle \mathbf{y}\rangle^{\bot})$ for any $0\neq\mathbf{c}\in C$.
\end{Proposition}

\begin{proof}
By the definition of $w_{p}$ and the function $\theta_{G}$, we have
\begin{align*}
  w_{p}(D) & =\big{|}\{0\leq i \leq n-1\,|\,\exists\,\mathbf{c}=(c_{0},c_{1},\cdots,c_{n-1})\in D,\,(c_{i},c_{i+1})\neq (0,0)\}\big{|} \\
   & =n-\big{|}\{0\leq i \leq n-1\,|\,\forall\,\mathbf{c}=(c_{0},c_{1},\cdots,c_{n-1})\in D,\,(c_{i},c_{i+1})= (0,0)\}\big{|} \\
   &  =n-\big{|}\{0\leq i \leq n-1\,|\,\forall\,\mathbf{y}\in \tilde{D}, \,\mathbf{y}G_{i}=\mathbf{y}G_{i+1}=0\,\,\}\big{|}\\
    &  =n-\big{|}\{0\leq i \leq n-1\,|\,\langle G_{i},G_{i+1}\rangle\subseteq\tilde{D}^{\bot}\}\big{|}  \\
    &=n-\sum_{V\in {\rm PG}^{\leq 2}(\tilde{D}^{\bot})}\big{|}\{0\leq i \leq n-1\,|\,\langle G_{i},G_{i+1}\rangle= V\}\big{|}\\
&=n-\sum_{V\in {\rm PG}^{\leq 2}(\tilde{D}^{\bot})}m_{G}(V)=n-\theta_{G}(\tilde{D}^{\bot}).
\end{align*}
In particular, when we take $D=\langle {\bf c}\rangle$ to be the $1$-dimensional subspace generated by the codeword ${\bf c}\in C$, then $w_p(\langle {\bf c}\rangle)=w_p({\bf c})=n-\theta_{G}(\langle \mathbf{y}\rangle^{\bot})$.
\end{proof}

%\begin{Remark}\label{n0}
%The dimension of $\langle G_{i},G_{i+1}\rangle$ could be $0$ for a generator matrix $G=(G_{0},\cdots,G_{n-1})$ of an $[n,k]$-linear code $C$. If $\langle G_{i},G_{i+1}\rangle=0$, we can construct a new linear code $\tilde{C}$ with a generator matrix $\tilde{G}=(G_{0},\cdots,G_{i},G_{i+2},\cdots,G_{n-1})$ and a linear isomorphism from $C$ to $\tilde{C}$ keeping the pair weight invariant. Without loss of generality, we will assume that $n_0(G)=0$ for a generator matrix $G=(G_{0},\cdots,G_{n-1})$ of a linear code $C$ in the rest of the paper.
% \end{Remark}

%Analogously, we give the definition of the pair $r$-equiweight code over $\mathbb{F}_q$ for $1\leq r\leq k-1$.
\begin{Definition}\label{CC}
Let $C$ be an $[n,k]$-linear code over $\mathbb{F}_q$ and $1\leq r\leq k-1$, we say that $C$ is a pair $r$-equiweight code if $d_{p}^{\,r}(C)=w_{p}(D)$ for any subspace $D$ of dimension $r$ of $C$.
\end{Definition}

\begin{Remark}
If $r=1$, the pair $1$-equiweight code is just the pair equiweight code. However, a Hamming equiweight code is not a pair equiweight code in general.
 \end{Remark}

\begin{Example}
Let $C_{1}$ be the linear code with a generator matrix $\left(\begin{array}{cccc}
                       1 & 0&1&0 \\
                        0 &1&0&1

\end{array}\right) $ over $\mathbb{F}_2$. Then $C_{1}$ is a pair equiweight code but not a Hamming equiweight code. Let $C_{2}$ be the linear code with a generator matrix $\left(\begin{array}{cccc}
                       1 & 1&0&0 \\
                        0 &1&1&0

\end{array}\right) $ over $\mathbb{F}_2$. Then $C_{2}$ is a Hamming equiweight code but not a pair equiweight code.
\end{Example}

The following proposition provides a method to construct a pair equiweight code from a Hamming equiweight code.

\begin{Proposition}\label{construct}
Let $C$ be an $[n,k]$-linear code over $\mathbb{F}_q$ with a generator matrix $G=(G_{0},\cdots,G_{n-1})$, and let $\hat{C}$ be a $[2n,k]$-linear code over $\mathbb{F}_q$ with a generator matrix $\hat{G}=(G_{0},O,\cdots,G_{n-1},O)$, where $O$ is the column zero vector of length $k$. Then for any $1\leq r \leq k-1$, $C$ is a Hamming $r$-equiweight code if and only if $\hat{C}$ is a pair $r$-equiweight code.

\end{Proposition}

\begin{proof}
Let $\varphi$ be a map from $C$ to $\hat{C}$ such that $\varphi(\mathbf{c})=(c_{0},0,\cdots,c_{n-1},0)\in \hat{C}$ for any $\mathbf{c}=(c_{0},\cdots,c_{n-1})\in C$. Then $\varphi$ is an $\mathbb{F}_q$-linear isomorphism and $w_{p}(\varphi(\mathbf{c}))=2w_{H}(\mathbf{c})$. The rest part of the proof is trivial.
\end{proof}

\section{Generalized pair weights of linear codes}
In this section, we give  general properties of generalized pair weights of linear codes. Some bounds about generalized pair weights of linear codes are obtained in this section.

We first give a characterization on the relationship between the generalized Hamming weight $w_{H}(D)$ and the generalized pair weight $w_{p}(D)$ for any $\mathbb{F}_q$-subspace $D$ of  $\mathbb{F}_q^{n}$. If $w_{H}(D)=n$, then $w_{p}(D)=n$. If  $w_{H}(D)<n$, we have the following lemma.

\begin{lem} \label{relationship-1}
Let $D$ be an $\mathbb{F}_q$-subspace of $\mathbb{F}_q^{n}$, and suppose $w_{H}(D)<n$. Assume that $$\chi_{H}(D)=\bigcup_{l=1}^{L}\{s_{l},s_{l}+1,\cdots,s_{l}+e_{l}\}\subseteq \mathbb{Z}/n\mathbb{Z}$$ and $ |s_{l}-s_{l-1}-e_{l-1}|\geq 2$ for $1\leq l\leq L$ where $s_{0}=s_{L}$ and $e_{0}=e_{L}$. Then $w_{p}(D)=w_{H}(D)+L$.
\end{lem}

\begin{proof}
If $i\in\chi_H(D)$, there exists $\mathbf{x}=(x_{0},\cdots,x_{n-1})\in D$ such that $x_{i}\neq 0$. Then the two pairs $(x_{i-1},x_{i})$ and $(x_{i},x_{i+1})$ both are not $(0,0)$ and $\{i-1,i\}\subseteq \chi_{p}(D)$. Here, when $i=0$, $i-1=n-1$. Hence $$\chi_{p}(D)=\cup_{l=1}^{L}\{s_{l}-1,s_{l},s_{l}+1,\cdots,s_{l}+e_{l}\}.$$ Since $ |s_{l}-s_{l-1}-e_{l-1}|\geq 2$, we have
$$
\{s_{l-1}-1,s_{l-1},s_{l-1}+1,\cdots,s_{l-1}+e_{l-1}\}\cap\{s_{l}-1,s_{l},s_{l}+1,\cdots,s_{l}+e_{l}\}=\varnothing
$$ for $1\leq l\leq L$, where $s_{0}=s_{L}$ and $e_{0}=e_{L}$. Hence $w_{p}(D)=|\chi_{p}(D)|=|\chi_H(D)|+L=w_{H}(D)+L$.
\end{proof}

\begin{Theorem}
Let $C$ be an $[n,k]$-linear code over $\mathbb{F}_q$. Then we have
\begin{description}
             \item[(a)] If $1\leq r\leq k-1$, or $r=k$  and $d_{H}^{\,k}(C)<n$, then $d_{H}^{\,r}(C)+1\leq d_{p}^{\,r}(C)\leq 2d_{H}^{\,r}(C)$.
             \item[(b)] If $r=k$  and $d_{H}^{\,k}(C)=n$ then $d_{p}^{\,k}(C)=n$.

  \end{description}

\end{Theorem}

\begin{proof}
%\begin{description}
(a) Suppose  $1\leq r\leq k-1$. Let $D$ be an $\mathbb{F}_q$-subspace of $C$ such that $\dim(D)=r$ and $d_{p}^{\,r}(C)=w_{p}(D)$. If $w_H(D)=|\chi_H(D)|=n$, then $$d_{p}^{\,r}(C)=w_p(D)=|\chi_{p}(D)|=n.$$ By Lemma~\ref{monotonicity}, there exists an $\mathbb{F}_q$-subspace $\tilde{D}$ of $C$ such that $\dim(\tilde{D})=r$ and $w_{H}(\tilde{D})=d_{H}^{\,r}(C)<n$. Then $n=w_p(D)=d_{p}^{\,r}(C)\leq w_p(\tilde{D})$ and hence
$$
w_p(D)=d_{p}^{\,r}(C)=w_p(\tilde{D})
$$
with $w_{H}(\tilde{D})=d_{H}^{\,r}(C)<n$. Therefore, without loss of generality, we can assume that $w_H(D)<n$. Then by Lemma~\ref{relationship-1}, we have $w_{p}(D)=w_{H}(D)+L$. Hence$$d_{p}^{\,r}(C)=w_{p}(D)=w_{H}(D)+L\geq w_{H}(D)+1\geq d_{H}^{\,r}(C)+1.$$
Let $E$ be an $\mathbb{F}_q$-subspace of $C$ such that $\dim(E)=r$ and $d_{H}^{\,r}(C)=w_{H}(E)$. Since $d_{H}^{\,r}(C)=w_{H}(E)<n$,  by Lemma~\ref{relationship-1},  we have $w_{p}(E)=w_{H}(E)+L_{1}$. Hence $d_{p}^{\,r}(C)\leq w_{p}(E)=w_{H}(E)+L_{1}\leq 2w_{H}(E)=2d_{H}^{\,r}(C)$.

If $r=k$ and $d_{H}^{\,k}(C)<n$,  we have $w_{p}(C)=w_{H}(C)+L_2$ by Lemma~\ref{relationship-1} since $|\chi_H(C)|=d_{H}^{\,k}(C)<n$. Hence
                    $$ d_{H}^{\,k}(C)+1=w_{H}(C)+1\leq d_{p}^{\,k}(C)=w_{p}(C)=w_{H}(C)+L_{2}\leq 2w_{H}(C)=2d_{H}^{\,k}(C).$$

(b) If $d_{H}^{\,k}(C)=|\chi_{H}(C)|=n$, then $d_{p}^{\,k}(C)= |\chi_{p}(C)|=n$.
\end{proof}

Note that, if $r=1$, we have $d_{H}^{\,1}(C)+1\leq d_{p}^{\,1}(C)\leq 2d_{H}^{\,1}(C)$, this is the usual relationship between the minimal pair weight and minimal Hamming weight of the linear code $C$.
\begin{Theorem}\label{pair monotonicity}
Let $C$ be an $[n,k]$-linear code over $\mathbb{F}_q$ with $n \ge 2$. Then we have $$2\leq d_{p}^{\,1}(C)<d_{p}^{\,2}(C)<\cdots<d_{p}^{\,k-1}(C) \leq d_{p}^{\,k}(C)\leq n.$$
\end{Theorem}

\begin{proof}
The inequlaity $d_{p}^{\,r}(C)\leq d_{p}^{\,r+1}(C)$ is trivial for $1\leq r \leq k-1$. For any subspace $D$ of dimension one of $C$ over $\mathbb{F}_q$, there exists $0\neq \mathbf{x}=(x_{0},\cdots,x_{n-1})\in D$ such that $x_{i}\neq 0$. Hence $w_{p}(D)\geq 2$ and $d_{p}^{\,1}(C)\geq2$.

For any $2\leq r \leq k-1$, by Lemma~\ref{monotonicity}, we have $d_{H}^{\,r}(C)<n$. Note that there exists a subspace $D$ of $C$ such that $\dim(D)=r$ and $w_p(D)=d_{p}^{\,r}(C)$ by the definition of the $r$-minimal pair weight of $C$. If $w_H(D)=|\chi_H(D)|=n$, then $$d_{p}^{\,r}(C)=w_p(D)=|\chi_{p}(D)|=n.$$ There exists an $\mathbb{F}_q$-subspace $\tilde{D}$ of $C$ such that $\dim(\tilde{D})=r$ and $w_{H}(\tilde{D})=d_{H}^{\,r}(C)<n$ by Lemma~\ref{monotonicity} again. Then $n=w_p(D)=d_{p}^{\,r}(C)\leq w_p(\tilde{D})$ and hence $w_p(D)=d_{p}^{\,r}(C)=w_p(\tilde{D})$. Therefore, without loss of generality, we can assume $w_H(D)<n$. Then there exists an index $i\in \chi_H(D)$ such that $i+1\not \in\chi_H(D)$, where $i+1$ is taken modulo $n$ when $i=n-1$. Let $\hat{D}=\{\mathbf{x}\in D\,|\,\mathbf{x}=(x_{0},\cdots,x_{n-1}),x_{i}=0\}$. We know that $\dim(\hat{D})=r-1$, $i\not\in \chi_H(\hat{D})$ and $\chi_H(\hat{D})\bigcup\{i\}=\chi_H(D)$. Hence $i\not\in \chi_{p}(\hat{D})$ and $i\in \chi_{p}(D)$. Therefore, $d_{p}^{\,r-1}(C)\leq|\chi_{p}(\hat{D})|<|\chi_{p}(D)|=d_{p}^{\,r}(C)$ for $2\leq r \leq k-1$.
\end{proof}

\begin{Remark}
 There exists a linear code $C$ of length $n$ such that $d_{p}^{\,k-1}(C)=d_{p}^{\,k}(C)=n$. For example, let $C$ be the linear code over $\mathbb{F}_2$ with generator matrix $\left(\begin{array}{ccc}
                       1 & 1&0 \\
                        0 &1&1

\end{array}\right)$. Then we have $d_{p}^{\,1}(C)=d_{p}^{\,2}(C)=3=n$.
\end{Remark}

\begin{Corollary}\label{equ condition}
Let $C$ be an $[n,k]$-linear code over $\mathbb{F}_q$ with $k\ge 2$. Then
\begin{description}
  \item[(a)] if $d_{p}^{\,k-1}(C)=d_{p}^{\,k}(C)$ then $d_{H}^{\,k}(C)=n$.
  \item[(b)] $d_{p}^{\,k-1}(C)=d_{p}^{\,k}(C)$ if and only if $\ d_{p}^{\,k-1}(C)=d_{p}^{\,k}(C)=n$.
\end{description}
\end{Corollary}

\begin{proof}
(a)\,\, Suppose otherwise that $d_{H}^{\,k}(C)<n$. Then there exists an index $i\in \chi_H(C)$ such that $i+1\not \in\chi_H(C)$, where the indices are taken modulo $n$. Let
$$\hat{C}=\{\mathbf{x}\in C\,|\,\mathbf{x}=(x_{0},\cdots,x_{n-1}),x_{i}=0\}.$$ We know $\dim(\hat{C})=k-1$, $i\not\in \chi_H(\hat{C})$ and $\chi_H(\hat{C})\bigcup\{i\}=\chi_H(C)$. Hence $i\not\in \chi_{p}(\hat{C})$ and $i\in \chi_{p}(C)$. Therefore, we have $$d_{p}^{\,k-1}(C)\leq|\chi_{p}(\hat{C})|<|\chi_{p}(C)|=d_{p}^{\,k}(C),$$ which is a contradiction.

(b)\,\, We only need to prove the necessity. By (a), $d_{H}^{\,k}(C)=n=|\chi_H(C)|$, hence $d_{p}^{\,k}(C)=|\chi_{p}(C)|=n$. Therefore, $\ d_{p}^{\,k-1}(C)=d_{p}^{\,k}(C)=n$.
\end{proof}

The claim ``$d_{H}^{\,k}(C)=n$ implies $d_{p}^{\,k-1}(C)=d_{p}^{\,k}(C)$" is not true in general. For example, let $C$ be a $[4,2]$-linear code over $\mathbb{F}_2$ with the generator matrix $G=\left(\begin{array}{cccc}
                       1 & 1&0&0 \\
                        0&0 &1&1

\end{array}\right) $. Then $d_{H}^{\,2}(C)=4$, $d_{p}^{\,1}(C)=3$ and $d_{p}^{\,2}(C)=4$.

 By using Theorem~\ref{pair monotonicity}, we can give a bound for generalized pair weight hierarchies $\{d_p^{\,1}(C), d_p^{\,2}(C),\cdots, d_p^{\,k}(C)\}$ and a relationship between this bound and MPDS codes defined in Remark~\ref{2.8}. For two real number sequences $\{a_1,a_2,\cdots,a_k\}$ and $\{b_1,b_2,\cdots,b_k\}$, $$\{a_1,a_2,\cdots,a_k\}\leq \{b_1,b_2,\cdots,b_k\}$$ means $a_i\leq b_i$ for any $1\leq i\leq k$.

\begin{Theorem}[Singleton Bound respect to generalized pair weights]\label{bound}
Let $C$ be an $[n,k]$-linear code over $\mathbb{F}_q$. Then
$$\{d_p^{\,1}(C), d_p^{\,2}(C),\cdots,d_p^{\,k-1}(C), d_p^{\,k}(C)\}\leq \{n-k+2,n-k+3,\cdots,n,n\}.$$ These bounds are met with equality everywhere if and only if $C$ is an MPDS code.
  \end{Theorem}

\begin{proof}
By Theorem~\ref{pair monotonicity}, for all $1\leq r\leq k-1$, we get $$d_{p}^{\,r}(C)\leq d_{p}^{\,r+1}(C)-1\leq\cdots \leq d_{p}^{\,k-1}(C)+r-k+1\leq n+r-k+1.$$ The remaining part of the proof is obvious.
\end{proof}

\begin{Example}
Let $C$ be the linear code with a generator matrix $\left(\begin{array}{cccc}
                       1 & 0&1&1 \\
                        0 &1&-1&1

\end{array}\right) $ over $\mathbb{F}_3$. Then we know that $C$ is an MPDS code by Proposition 4.1 of \cite{C}. On  the other hand, we have $d_p^{\,i}(C)=4$ for any $1\leq i \leq 2$ by directly calculating.

\end{Example}

By using Theorem~\ref{bound}, it is easy to know that a linear code is not an MPDS code when there exists an index  $r$ such that $1\leq r\leq k-1$ and $d_{p}^{\,r}(C)<n-k+r+1$. Then we prove that the definition of the LDP for the pair weight which is essentially the same as the generalized pair weight hierarchy.

\begin{Theorem}\label{essent}
Assume the notations are given above. Then $d_p^r(C)=m_r(C)$ for all $1\leq r \leq k$.
\end{Theorem}

\begin{proof}
Assume $1\leq r\leq k$. There is a subset $J_0$  of $\{0,1,\cdots,n-1\}$ such that $\dim(C_{J_0})=r$ and $|J_0|=m_r(C)$ by the definition of $m_r(C)$. By the definition of $d_p^{\,r}(C)$, we have
\begin{equation}\label{2.2}
  d_p^{\,r}(C)\leq w_p(C_{J_0})\leq |J_0|=m_r(C).
\end{equation}

On the other hand, there is a subspace $D$ of $C$ such that $\dim(D)=r$ and $d_p^{\,r}(C)=w_p(D)$ by the definition of $d_p^{\,r}(C)$. Assume $J_1=\chi_p(D)$, then $D\leq C_{J_1}$.

If $D= C_{J_1}$, then $m_r(C)\leq |J_1|=w_p(D)=d_p^{\,r}(C)$. Hence $d_p^{\,r}(C)=m_r(C)$.

If $D\subsetneqq C_{J_1}$, we get $$\hat{r}=\dim(C_{J_1})>\dim(D)=r.$$ Hence $d_p^{\,\hat{r}}(C)\leq w_p(C_{J_1})\leq |J_{1}|=w_p(D)=d_p^{\,r}(C)$. By Theorem~\ref{pair monotonicity}, we get $r=k-1=\hat{r}-1$ and $d_p^{\,k-1}(C)=d_p^{\,k}(C)$. Then $d_{p}^{\,k-1}(C)=d_{p}^{\,k}(C)=n$ by Corollary~\ref{equ condition}. Hence $d_p^{\,r}(C)=m_r(C)$ by Inequality~\ref{2.2} and $m_r(C)\leq n$.
\end{proof}

\section{ Pair $r$-equiweight codes}

In this section, we study pair $r$-equiweight codes. Before we provide our main theorems in this section, we give some notions and a key lemma.

Recall that $n_{r, k}$ is the number of all subspaces of dimension $r$ of a vector space of dimension $k$.
Let ${\rm PG}^{r}(\mathbb{F}_q^{k})=\{V_{1}^r,V_{2}^r,\cdots,V_{n_{r,k}}^r \}$ be the set of all subspaces of dimension $r$ of $\mathbb{F}_q^k$. There is a bijection between ${\rm PG}^{k-r}(\mathbb{F}_q^{k})$ and ${\rm PG}^{r}(\mathbb{F}_q^{k})$, which is defined by
$$
{\rm PG}^{k-r}(\mathbb{F}_q^{k})\to {\rm PG}^{r}(\mathbb{F}_q^{k}), V^{k-r}\mapsto(V^{k-r})^\bot,\,\, \forall\,\,  V^{k-r}\in {\rm PG}^{k-r}(\mathbb{F}_q^{k}).
$$
Hence $n_{r, k}=n_{k-r, k}$. For convenience, if $\frac{k}{2}<r\leq k$, we assume $${\rm PG}^{r}(\mathbb{F}_q^{k})=\{V_{1}^{r}=(V_{1}^{k-r})^\bot,V_{2}^{r}=(V_{2}^{k-r})^\bot,\cdots,V_{n_{r, k}}^{r}=(V_{n_{r, k}}^{k-r})^\bot \}.
$$

Let $\mathbb{Q}$ be the rational number field. For $0\leq r \leq s\leq k$, let $T_{r,s}$ be a matrix in $M_{n_{r,k}\times n_{s,k}}(\mathbb{Q})$ such that
$$T_{r,s}=(t_{ij})_{n_{r,k}\times n_{s,k}},\,\,\,\,\, \mbox{where}\,\, t_{ij} = \left\{ \begin{array}{ll}
1,  & \textrm{if $V_i^r\subseteq V_j^s ;$}\\
0,  & \textrm{if $V_i^r\nsubseteq V_j^s .$}
\end{array} \right.$$
Let $A^T$ denote the transpose matrix of the matrix $A$. Let $J_{m\times n}$ be the $m\times n$ matrix with all entries being $1$. i.e, $J_{m\times n}=\left(\begin{array}{cccc}
                       1 & \cdots&1\\
                       \vdots&\ddots&\vdots\\
                        1&\cdots&1
\end{array}\right) $. In particular, $J_{1\times n}={\bf 1}=(1,\cdots, 1)$.

\begin{lem} \label{T1}
Assume the notations are given above, and let $k\ge 2$. Then
\begin{description}
  \item[(a)] The sum of all rows of $T_{r,s}$ is the constant row vector $n_{r,s}{\bf 1}$.

  \item[(b)] The matrix $T_{1,k-1}$ is an invertible matrix and $T_{1,k-1}^{-1}=\frac{1}{q^{k-2}}(T_{1,k-1}-\frac{q^{k-2}-1}{q^{k-1}-1}J_{n_{1,k}\times n_{1,k}})$. The sum of all rows of $T_{1,k-1}^{-1}$ is a constant row vector.

  \item[(c)] $T_{r,k-1}T_{1,k-1}=(q^{k-r-1})T_{1,r}^{T}+\frac{q^{k-r-1}-1}{q-1}J_{n_{r,k}\times n_{1,k}}$ and $$T_{r,k-1}T_{1,k-1}^{-1}=\frac{1}{q^{r-1}}T_{1,r}^{T}-\frac{q^{r-1}-1}{q^{r-1}(q^{k-1}-1)}J_{n_{r, k}\times n_{1, k}},\,\,  \mbox{for}\,\, k\ge r+1.$$

  \item[(d)] $T_{r,s}T_{s,z}=n_{s-r,z-r}T_{r,z}$ for $1\leq r\leq s\leq z\leq k$.
\end{description}
\end{lem}

\begin{proof}
(a) Since the number of all subspaces of dimension $r$ of $V_{i}^s$ is $n_{r, s}$ for any $1\leq i\leq n_{s,k}$, we know that the sum of the rows of $T_{r,s}$ is the constant row vector $n_{r,s}{\bf 1}$.

(b) By (a), we know $J_{n_{1, k}\times n_{1, k} }T_{1,k-1}=n_{1,k-1}J_{n_{1, k}\times n_{1, k} }$. Since $V_{i}^{k-1}=(V_{i}^1)^\bot$ for any $1\leq i\leq n_{1,k-1}$, we get $T_{1,k-1}=T_{1,k-1}^{T}$. Then $$T_{1,k-1}T_{1,k-1}=T_{1,k-1}^{T}T_{1,k-1}=(b_{ij})_{n_{1,k}\times n_{1,k}},\,\,\,b_{ij} = \left\{ \begin{array}{ll}
n_{1,k-1},  & \textrm{if $i=j ;$}\\
n_{1,k-2},  & \textrm{if $i\neq j .$}
\end{array} \right.  , $$ since $b_{ij}$ is the number of all subspace of dimension one of $V^{k-1}_{i}\bigcap V^{k-1}_{j}$ for $1\leq i,j\leq n_{1, k}$. Then $$\frac{1}{n_{1,k-1}-n_{1,k-2}}(T_{1,k-1}-\frac{n_{1,k-2}}{n_{1,k-1}}J_{n_{1, k}\times n_{1, k} })T_{1,k-1}$$ is identity matrix. Hence $T_{1,k-1}$ is an invertible matrix and
\begin{align*}
  T_{1,k-1}^{-1} & =\frac{1}{n_{1,k-1}-n_{1,k-2}}(T_{1,k-1}-\frac{n_{1,k-2}}{n_{1,k-1}}J_{n_{1, k}\times n_{1, k} }) \\
   & =\frac{1}{q^{k-2}}(T_{1,k-1}-\frac{q^{k-2}-1}{q^{k-1}-1}J_{n_{1,k}\times n_{1,k}}).
\end{align*}

Since the sum of all rows of $T_{1,k-1}$ and the sum of all rows of $J_{n_{1,k}\times n_{1,k}}$ are constant row vectors,  the sum of all rows of $T_{1,k-1}^{-1}$ is a constant row vector.

(c) By the definition of $T_{1,k-1}$ and $T_{r,k-1}$, we have $$T_{r,k-1}T_{1,k-1}=T_{r,k-1}T_{1,k-1}^{T}=(c_{ij})_{n_{r,k}\times n_{1,k}}$$
such that $$c_{ij}=\big{|}\{V^{k-1}_{s}|1\leq s \leq n_{1,k},\langle V^1_{j}, V^r_{i}\rangle\leq V^{k-1}_{s} \}\big{|}$$ for $1\leq i \leq n_{r,k}$ and $1\leq j \leq n_{1, k}$. If $V^1_{j}\leq V^r_{i}$,
      \begin{align*}
        c_{ij} & =\big{|}\{V^{k-1}_{s}|1\leq s \leq n_{1, k},V^r_{i}\leq V^{k-1}_{s} \}\big{|} \\
         & =\big{|}\{ M \,|\,M\leq \mathbb{F}_{q}^{k}/ V^{r}_{i},\,\dim(M)=k-r-1  \}\big{|}  \\
         &=n_{1,k-r}.
      \end{align*}
If $V^1_{j}\nsubseteq V^r_{i}$,
       \begin{align*}
        c_{ij} & =\big{|}\{V^{k-1}_{s}\,|\,1\leq s \leq n_{1, k},\langle V^1_{j}, V^r_{i}\rangle\subseteq V^{k-1}_{s} \}\big{|}\\
         & =\big{|}\{ M \,|\,M\leq \mathbb{F}_{q}^{k}/ \langle V^1_{j}, V^r_{i}\rangle,\,\dim(M)=k-r-2  \}\big{|}  \\
         &=n_{1,k-r-1}.
      \end{align*}

Hence $c_{ij} = \left\{ \begin{array}{ll}
n_{1,k-r},  & \textrm{if $V^1_{j}\subseteq V^r_{i} ;$}\\
n_{1,k-r-1},  & \textrm{if $V^1_{j}\nsubseteq V^r_{i} .$}
\end{array} \right.$ And

\begin{align*}
 T_{r,k-1}T_{1,k-1} & =(n_{1,k-r}-n_{1,k-r-1})T_{1,r}^{T}+n_{1,k-r-1}J_{n_{r, k}\times n_{1, k}} \\
   & =q^{k-r-1}T_{1,r}^{T}+\frac{q^{k-r-1}-1}{q-1}J_{n_{r, k}\times n_{1, k}}.
\end{align*}

Since $T_{r,k-1}J_{n_{1, k}\times n_{1, k}}=n_{1,k-r}J_{n_{r, k}\times n_{1, k}}$, we have
\begin{align*}
  T_{r,k-1}T_{1,k-1}^{-1} & =\frac{1}{n_{1,k-1}-n_{1,k-2}}T_{r,k-1}(T_{1,k-1}-\frac{n_{1,k-2}}{n_{1,k-1}}J_{n_{1, k}\times n_{1, k}}) \\
   & =\frac{1}{q^{r-1}}T_{1,r}^{T}-\frac{q^{r-1}-1}{q^{r-1}(q^{k-1}-1)}J_{n_{r, k}\times n_{1, k}}.
\end{align*}

(d)  By the definition of $T_{r,s}$ and $T_{s,z}$, we have $T_{r,s}T_{s,z}=(d_{ij})_{n_{r,k}\times n_{z,k}}$ such that $$d_{ij}=|\{V^s_{l}\,|\,1\leq l \leq n_{s,k},\,V^r_{i}\leq V^s_{l}\leq V^z_{j}\}|$$ for $1\leq i \leq n_{r,k}$ and $1\leq j \leq n_{z,k}$. If $V^r_{i}\subseteq V^z_{j}$,
      \begin{align*}
        d_{ij} & =\big{|}\{V^s_l\,|\,1\leq l \leq n_{s,k},\, V^r_{i}\leq U\leq V^z_{j}\}\big{|} \\
         & =\big{|}\{U\,|\, U\leq V^z_{j}/V^r_{i},\,\dim(U)=s-r \}\big{|} \\
         &=n_{s-r,z-r}.
      \end{align*}
       If $V^r_{i}\nsubseteq V^z_{j}$, $d_{ij}=0$. Hence $T_{r,s}T_{s,z}=(d_{ij})_{n_{r,k}\times n_{z,k}}=n_{s-r,z-r}T_{r,z}$.
\end{proof}

It is easy to see that when $k=1$, any $[n,1]$-linear code is a pair equiweight code. In the following we assume $k\ge 2$, and study pair equiweight linear codes.

\begin{Theorem}\label{2 condition}
Assume the notations are given above. Let $C$ be an $[n,k]$-linear code over $\mathbb{F}_q$ with a generator matrix $G=(G_{0},\cdots,G_{n-1})$ for $k\ge 2$. Then $C$ is a pair equiweight code if and only if $\sum\limits_{V\in \Omega_i  }\frac{1}{|V|}m_G(V)$ is constant for any $1\leq i \leq n_{1,k}$, where $s=\min\{2,k-1\}$ and $\Omega_i=\{V\in {\rm PG}^{\leq s}(\mathbb{F}_q^k)\,|\,V^1_i\subseteq V\}$.
\end{Theorem}

\begin{proof}
For $0\leq r \leq 2$, let $\Delta_{r}=(m_{G}(V^r_{1}),m_{G}(V^r_{2}),\cdots,m_{G}(V^r_{n_{r,k}}))$, and let $$\Gamma_{k-1}=(\theta_{G}(V^{k-1}_{1}),\theta_{G}(V^{k-1}_{2}),\cdots,\theta_{G}(V^{k-1}_{n_{1,k-1}})).$$ Assume $s=\min\{2,k-1\}$. By the definition of $\theta_{G}$, we can verify that
\begin{equation}\label{a}
  \Gamma_{k-1}=(\sum\limits_{W\in {\rm PG}^{\le 2}(V^{k-1}_{1})}m_G(W),\cdots,\sum\limits_{W\in {\rm PG}^{\le 2}(V^{k-1}_{n_{1,k-1}})}m_G(W))=\sum_{r=0}^{s}\Delta_{r}T_{r,k-1}.
\end{equation}
By Lemma~\ref{T1} (b), the above equation is
\begin{equation}\label{aaaaa}
  \Gamma_{k-1}-m_G(\textbf{0})\textbf{1}=(\sum_{r=1}^{s}\Delta_{r}T_{r,k-1}T_{1,k-1}^{-1})T_{1,k-1}.
\end{equation}
By Lemma~\ref{T1} (c), the element in the $i$th position of the vector $$\sum_{r=1}^{s}\Delta_{r}T_{r,k-1}T_{1,k-1}^{-1}=\sum_{r=1}^{s}\Delta_{r}(\frac{1}{q^{r-1}}T_{1,r}^{T}-\frac{q^{r-1}-1}{q^{r-1}(q^{k-1}-1)}J_{n_{r, k}\times n_{1, k}})$$
is
\begin{align}\label{abc}
   & m_G(V_i^1)+\sum_{r=2}^s(\frac{1}{q^{r-1}}\sum_{ V^r\in {\rm PG}^r(\mathbb{F}_q^k),V^1_i\subseteq V^r }m_G(V^r)-\frac{q^{r-1}-1}{q^{r-1}(q^{k-1}-1)}\sum_{V^r\in {\rm PG}^r(\mathbb{F}_q^k)}m_G(V^r)) \notag \\ \notag
   & =m_G(V_i^1)+\sum_{r=2}^s\sum_{V^r\in {\rm PG}^r(\mathbb{F}_q^k),V^1_i\subseteq V^r }\frac{1}{q^{r-1}}m_G(V^r)-\sum_{r=2}^s\sum_{V^r\in {\rm PG}^r(\mathbb{F}_q^k)}\frac{q^{r-1}-1}{q^{r-1}(q^{k-1}-1)}m_G(V^r)\\
   & =q\sum_{V\in \Omega_i  }\frac{1}{|V|}m_G(V)-\sum_{r=2}^s\sum_{V^r\in {\rm PG}^r(\mathbb{F}_q^k)}\frac{q^{r-1}-1}{q^{r-1}(q^{k-1}-1)}m_G(V^r),
\end{align}
where $\Omega_i=\{V\in {\rm PG}^{\leq s}(\mathbb{F}_q^k)\,|\,V^1_i\subseteq V\}$.

Now suppose $\sum\limits_{V\in \Omega_i  }\frac{1}{|V|}m_G(V)$ is constant for all $1\leq i\leq n_{1, k}$. Then $$q\sum_{V\in \Omega_i  }\frac{1}{|V|}m_G(V)-\sum_{r=2}^s\sum_{V^r\in {\rm PG}^r(\mathbb{F}_q^k)}\frac{q^{r-1}-1}{q^{r-1}(q^{k-1}-1)}m_G(V^r)$$ is constant for all $1\leq i\leq n_{1, k}$ and $\sum_{r=1}^{s}\Delta_{r}T_{r,k-1}T_{1,k-1}^{-1}$ is a constant vector by Equation~\ref{abc}. Since the sum of all rows of $T_{1,k-1}$ is a constant row vector by Lemma~\ref{T1} (a), we get $$\Gamma_{k-1}-m_G(\textbf{0})\textbf{1}$$ and $\Gamma_{k-1}$ are constant vectors by Equation~\ref{aaaaa}. Then $\theta_{G}$ is a constant function.

Since the $\mathbb{F}_q$-linear map $\phi:\,\mathbb{F}_q^{k}\rightarrow C$ such that $\phi(\mathbf{y})=\mathbf{y}G$ for any $\mathbf{y}\in \mathbb{F}_q^{k}$ is a linear isomorphism, there is nonzero vector $\mathbf{y}$ in $\mathbb{F}_q^{k}$ such that $\mathbf{c}=\phi(\mathbf{y} )$ for any nonzero codeword $\mathbf{c}$ in $C$. By Proposition~\ref{r pair weight}, we have $$w_{p}(\mathbf{c})= n-\theta_{G}(\langle\mathbf{y}\rangle^{\bot}).$$
Hence $C$ is a pair equiweight code.

On the contrary, suppose $C$ is a pair equiweight code. Then $\Gamma_{k-1}$, $\Gamma_{k-1}-m_G(\textbf{0})\textbf{1}$ and $\sum\limits_{r=1}^{s}\Delta_{r}T_{r,k-1}T_{1,k-1}^{-1}$ are all constant vectors by Proposition~\ref{r pair weight}, Lemma~\ref{T1} (b) and Equation~\ref{aaaaa}. Then $\sum\limits_{V\in \Omega_i  }\frac{1}{|V|}m_G(V)$ is constant for all $1\leq i\leq n_{1, k}$, since the element in the $i$th position of the vector $ \sum\limits_{r=1}^{s}\Delta_{r}T_{r,k-1}T_{1,k-1}^{-1}$ is $$q\sum_{V\in \Omega_i  }\frac{1}{|V|}m_G(V)-\sum_{r=2}^s\sum_{V^r\in {\rm PG}^r(\mathbb{F}_q^k)}\frac{q^{r-1}-1}{q^{r-1}(q^{k-1}-1)}m_G(V^r).$$
\end{proof}

%\begin{Example}
%Let $C$ be the linear code with generator matrix $\left(\begin{array}{ccccc}
                       %1 & 1&1&0&0\\
                        %0 &0&1&0&1

%\end{array}\right) $ over $\mathbb{F}_2$, then $C$ is a pair equiweight code and the value of the pair weight of $C$ is $4$.
%\end{Example}

In particular, if $k\ge 3$ and the function $m_{G}$ is constant restricted on ${\rm PG}^{2}(\mathbb{F}_q^{k})$, then we have the following corollary.

%\begin{Corollary}\label{construct2}
%Assume the notations given above. Let $C$ be an $[n,k]$-linear code over $\mathbb{F}_q$ with a generator matrix $G=(G_{0},\cdots,G_{n-1})$ and $2\leq k$. There is a linear code $\hat{C}$ over $\mathbb{F}_q$ with a generator matrix $\hat{G}=[G_{0},\cdots,G_{n-1},\alpha_{1},\alpha_{2},\cdots,\alpha_{s}]$ for some $\alpha_{i}$ in $\mathbb{F}_q^{k}$ as column vector for $1\leq i \leq s$ such that $\hat{C}$ is an equidistant pair weight code.
%\end{Corollary}

\begin{Corollary}\label{n dim}
Assume the notations are given above. Let $C$ be an $[n,k]$-linear code over $\mathbb{F}_q$ with a generator matrix $G=(G_{0},\cdots,G_{n-1})$. Then $C$ is a pair equiweight code if and only if the function $m_{G}$ restricted on ${\rm PG}^{1}(\mathbb{F}_q^{k})$ is a constant.
\end{Corollary}

\begin{proof}
Suppose the function $m_{G}$ is constant function on ${\rm PG}^{2}(\mathbb{F}_q^{k})$ with value $a\in\mathbb{N}$, then $$\sum_{V\in \Omega_i  }\frac{1}{|V|}m_G(V)=m_{G}(V^1_{i})+\frac{|\Omega_i|-1}{q^2}a,$$ where $\Omega_i=\{V\in {\rm PG}^{\leq 2}(\mathbb{F}_q^k)\,|\,V^1_i\subseteq V\}$. Hence the function $m_{G}$ for $G$ is a constant function on ${\rm PG}^{1}(\mathbb{F}_q^{k})$ if and only if $$\sum_{V\in \Omega_i  }\frac{1}{|V|}m_G(V)$$ is constant for all $1\leq i \leq n_{1, k}$, if and only if $C$ is a pair equiweight code by statement (b) in Theorem~\ref{2 condition}.
%Suppose the function $\hat{m}_{G}$ for $G$ is a constant function with value $b\in\mathbb{N}$, then
%$\hat{m}_{G}^{k-1}(\langle \mathbf{y}\rangle^{\bot})=\sum_{V\in PG^{1}(\langle \mathbf{y}\rangle^{\bot})}\hat{m}_{G}(V)=\frac{(q^{k-1}-1)}{(q-1)}b$ and $w_{p}(\mathbf{x})=n-\frac{(q^{k-1}-1)}{(q-1)}b-\frac{(q^{k-1}-1)(q^{k-1}-q)}{(q^{2}-1)(q^{2}-q)}a$. Hence $C$ is a equidistant pair weight code.
%Suppose $C$ is a equidistant pair weight code. Then $w_{p}(\mathbf{x})=n-\hat{m}_{G}^{k-1}(\langle \mathbf{y}\rangle^{\bot})-\frac{(q^{k-1}-1)(q^{k-1}-q)}{(q^{2}-1)(q^{2}-q)}a$ is invariant for any $0\neq\mathbf{x}\in C$ and the function $\hat{m}_{G}^{k-1}$ is constant a function from $PG^{k-1}(\mathbb{F}_q^{k})$ to $\mathbb{N}=\{0,1,2,\cdots\}$. By Lemma~\ref{r function}, the function $\hat{m}_{G}$ is a constant function.
\end{proof}

We can get the  following example by using Corollary~\ref{n dim}.
\begin{Example}
Let $\alpha_{1}=(011000001010001001011),\alpha_{2}=(001011000001011001010)$ and $\alpha_{3}=(000001011001001011001)\in\mathbb{F}_2^{21}$.
Let $C$ be the linear code with a generator matrix $\left(\begin{array}{c}
                       \alpha_{1}\\
                        \alpha_{2}\\
                        \alpha_{3}

\end{array}\right) $ over $\mathbb{F}_2$. Then $C$ is a pair equiweight code and the value of the pair weight of $C$ is $14$ by Corollary~\ref{n dim}. Also we can directly calculate to get the following table such that the first and the second column are non-zero vectors in $C$ and the third column is the pair weight of the vector which is at the same row.

$$\begin{tabular}{|c|c|c|}
  \hline
  $\alpha_{1}$ & (011000001010001001011) & 14 \\
  \hline
  $\alpha_{2}$ &(001011000001011001010)& 14\\
  \hline
  $\alpha_{3}$ &(000001011001001011001)& 14\\
  \hline
  $\alpha_{1}+\alpha_{2}$ & (010011001011010000001) & 14 \\
  \hline
  $\alpha_{1}+\alpha_{3}$ & (011001010011000010010) & 14 \\
  \hline
  $\alpha_{2}+\alpha_{3}$ & (001010011000010010011) & 14 \\
  \hline
  $\alpha_{1}+\alpha_{2}+\alpha_{3}$ & (010010010010011011000) & 14 \\
  \hline
\end{tabular}$$

\end{Example}

In the  next theorem, we obtain a necessary condition and a sufficient condition of that $C$ is a pair $r$-equiweight code.
\begin{Theorem}\label{k-1 dim}
Assume the notations are given above. Let $C$ be an $[n,k]$-linear code over $\mathbb{F}_q$ with a generator matrix $G=(G_{0},\cdots,G_{n-1})$, $k\ge 2$ and $1\leq r \leq k-1$.
\begin{description}
\item[(a)] If $C$ is a pair $r$-equiweight code, then $$\sum_{V\in \Omega_i  }\frac{n_{k-r-\dim(V),k-1-\dim(V)}}{|V|}m_G(V)$$ is constant for any $1\leq i \leq n_{1,k}$, where $s=\min\{2,k-r\}$, $\Omega_i=\{V\in {\rm PG}^{\leq s}(\mathbb{F}_q^k)\,|\,V^1_i\subseteq V\}$.

\item[(b)] When $r=k-1$,  $C$ is a pair $r$-equiweight code if and only if the function $m_{G}$ restricted on ${\rm PG}^{1}(\mathbb{F}_q^{k})$ is a constant.

\item[(c)] When $2\leq r\leq k-2$, if $m_{G}(V^2_{i})+\frac{1}{n_{1,k-r-1}}\sum_{V^1\in {\rm PG}^{1}(V^2_{i})}m_{G}(V^1)$  is constant for $1\leq i \leq n_{2,k}$, then $C$ is a pair $r$-equiweight code.
\end{description}
\end{Theorem}

\begin{proof}
(a) Since the $\mathbb{F}_q$-linear map $\phi:\,\mathbb{F}_q^{k}\rightarrow C$ such that $\phi(\mathbf{y})=\mathbf{y}G$ for any $\mathbf{y}\in \mathbb{F}_q^{k}$ is a linear isomorphism, there is an unique $\mathbb{F}_q$-subspace $\tilde{D}$ of $\mathbb{F}_q^{k}$ such that $D=\phi(\tilde{D} )$ for any $\mathbb{F}_q$-subspace $D$ with $\dim(D)=r$ of $C$. By Proposition~\ref{r pair weight}, we have $$w_{p}(D)= n-\theta_{G}(\tilde{D}^{\bot}).$$
 Let $\Delta_{l}=(m_{G}(V^l_{1}),m_{G}(V^l_{2}),\cdots,m_{G}(V^l_{n_{l,k}}))$ for $0\leq l \leq s$ and  $$\Gamma_{k-r}=(\theta_{G}(V^{k-r}_{1}),\theta_{G}(V^{k-r}_{2}),\cdots,\theta_{G}(V^{k-r}_{n_{r,k-1}})).$$ By the definition of the function $\theta_{G}$, we get
\begin{equation}\label{a}
  \Gamma_{k-r}=\sum_{l=0}^{s}\Delta_{l}T_{l,k-r}
\end{equation}
and
$$\Gamma_{k-r}-m_G(\textbf{0})\textbf{1}=\sum_{l=1}^{s}\Delta_{l}T_{l,k-r}.$$
By Lemma~\ref{T1} (b) and (d), we have \begin{align}\label{asa1}
  (\Gamma_{k-r}-m_G(\textbf{0})\textbf{1})T_{k-r,k-1}T_{1,k-1}^{-1} & =\sum_{l=1}^{s}\Delta_{l}T_{l,k-r}T_{k-r,k-1}T_{1,k-1}^{-1} \notag\\
   & =\sum_{l=1}^{s}n_{k-l-r,k-l-1}\Delta_{l}T_{l,k-1}T_{1,k-1}^{-1}.
\end{align}
Also we know the element in the $i$th position of the vector $$\sum_{l=1}^{s}n_{k-l-r,k-l-1}\Delta_{l}T_{l,k-1}T_{1,k-1}^{-1}=\sum_{l=1}^{s}n_{k-l-r,k-l-1}\Delta_{l}(\frac{1}{q^{l-1}}T_{1,l}^{T}-\frac{q^{l-1}-1}{q^{l-1}(q^{k-1}-1)}J_{n_{l, k}\times n_{1, k}})$$ is

%$ n_{k-1-r,k-2}m_G(V_i^1)+\sum_{l=2}^s(\frac{n_{k-l-r,k-l-1}}{q^{l-1}}\sum_{ V^l\in PG^l(\mathbb{F}_q^k),V^1_i\subseteq V^l }m_G(V^l))-\sum_{l=2}^s(n_{k-l-r,k-l-1}\frac{q^{l-1}-1}{q^{l-1}(q^{k-1}-1)}\sum_{V^l\in PG^l(\mathbb{F}_q^k)}m_G(V^l)) $
    %$n_{k-1-r,k-2}m_G(V_i^1)+\sum_{l=2}^s\sum_{ V^l\in PG^l(\mathbb{F}_q^k),V^1_i\subseteq V^l }\frac{n_{k-l-r,k-l-1}}{q^{l-1}}m_G(V^l)-\sum_{l=2}^s(n_{k-l-r,k-l-1}\frac{q^{l-1}-1}{q^{l-1}(q^{k-1}-1)}\sum_{V^l\in PG^l(\mathbb{F}_q^k)}m_G(V^l))\\$
\begin{align}\label{aaaaaa1}
   q\sum_{V\in \Omega_i  }\frac{n_{k-r-\dim(V),k-1-\dim(V)}}{|V|}m_G(V)-\sum_{l=2}^s\sum_{V^l\in {\rm PG}^l(\mathbb{F}_q^k)}n_{k-l-r,k-l-1}\frac{q^{l-1}-1}{q^{l-1}(q^{k-1}-1)}m_G(V^l)
\end{align}
by Lemma~\ref{T1} (c), where $\Omega_i=\{V\in {\rm PG}^{\leq s}(\mathbb{F}_q^k)\,|\,V^1_i\subseteq V\}$.

Now suppose $C$ is a pair $r$-equiweight code. Then $\Gamma_{k-1}$, $\Gamma_{k-1}-m_G(\textbf{0})\textbf{1}$ and $$\sum_{l=1}^{s}n_{k-l-r,k-l-1}\Delta_{l}T_{l,k-1}T_{1,k-1}^{-1}$$ are both constant vectors by Proposition~\ref{r pair weight}, Lemma~\ref{T1} (a) and (b), and Equation~\ref{asa1}. Then $$\sum_{V\in \Omega_i  }\frac{n_{k-r-\dim(V),k-1-\dim(V)}}{|V|}m_G(V)$$ is constant for all $1\leq i\leq n_{1, k}$ by Equation~\ref{aaaaaa1}.

(b) When $r=k-1$, then  $s=1$ and $$\sum_{V\in \Omega_i  }\frac{n_{k-r-\dim(V),k-1-\dim(V)}}{|V|}m_G(V)=\frac{1}{q}m_G(V_i^1).$$
Then we use (a).

(c)When $2\leq r \leq k-2$, then  $s=2$ and Equation~\ref{a} is
\begin{equation}\label{a1}
  \Gamma_{k-r}-m_G(\textbf{0})\textbf{1}=\Delta_{1}T_{1,k-r}+\Delta_{2}T_{2,k-r}=(\frac{1}{n_{1,k-r-1}}\Delta_{1}T_{1,2}+\Delta_{2})T_{2,k-r}.
\end{equation}

Since the element in the $i$th position of the vector $\frac{1}{n_{1,k-r-1}}\Delta_{1}T_{1,2}+\Delta_{2}$ is $$m_{G}(V^2_{i})+\frac{1}{n_{1,k-r-1}}\sum_{V^1\in {\rm PG}^{1}(V^2_{i})}m_{G}(V^1)$$ which is constant for $1\leq i \leq n_{2,k}$ as assumption, we have $$\frac{1}{n_{1,k-r-1}}\Delta_{1}T_{1,2}+\Delta_{2},$$ $\Gamma_{k-r}-m_G(\textbf{0})\textbf{1}$ and $\Gamma_{k-r}$ are both constant vectors by Lemma~\ref{T1} (a). Hence $C$ is a pair $r$-equiweight code by Proposition~\ref{r pair weight}.

\end{proof}

\section{Linear isomorphisms preserving pair weights}

MacWilliams \cite{M} and later Bogart, Goldberg, and Gordon
\cite{BGG} proved that every linear isomorphism preserving Hamming weights between two linear codes over finite fields can be induced by a monomial matrix. Unfortunately,  a linear isomorphism induced by a permutation matrix may not preserve  pair weights of linear codes. In this section,  we  obtain a necessary and sufficient condition for a linear isomorphism preserving  pair weights between two linear codes.

Let $C$ and $\tilde{C}$ be two $[n,k]$-linear code over $\mathbb{F}_q$ and $G=\left(\begin{array}{c}
                       {\bf g}_{1}\\
                        \cdots \\
                        {\bf g}_{k}

\end{array}\right)$ be a generator matrix of $C$ for some ${\bf g}_{i}\in \mathbb{F}_q^{n}$. Let $\varphi$ be an $\mathbb{F}_q$-linear isomorphism from $C$ to $\tilde{C}$. Then $\tilde{G}=\left(\begin{array}{c}
                      \varphi( {\bf g}_{1})\\
                        \cdots \\
                        \varphi({\bf g}_{k})

\end{array}\right)$ is a generator matrix of $\tilde{C}$. Before we give a necessary and sufficient condition of that $w_p(\mathbf{c})=w_p(\varphi(\mathbf{c}))$ for any $\mathbf{c}\in C$, we need following theorem.

\begin{Theorem}\label{6.1}
Assume the notations are given above. Then $w_p(\mathbf{c})-w_p(\varphi(\mathbf{c}))$ is constant for any nonzero $\mathbf{c}\in C$ if and only if $\sum\limits_{V\in \Omega_i  }\frac{1}{|V|}(m_G(V)-m_{\tilde{G}}(V))$ is constant for any $1\leq i \leq n_{1,k}$, where $s=\min\{2,k-1\}$, $\Omega_i=\{V\in {\rm PG}^{\leq s}(\mathbb{F}_q^k)\,|\,V^1_i\subseteq V\}$.

\end{Theorem}

\begin{proof}
Let $\phi$ be the $\mathbb{F}_q$-linear isomorphism from $\mathbb{F}_q^{k}$ to $ C$ such that $\phi(\mathbf{y})=\mathbf{y}G$ for any $\mathbf{y}\in \mathbb{F}_q^{k}$. And let $\tilde{\phi}$ be the $\mathbb{F}_q$-linear isomorphism from $\mathbb{F}_q^{k}$ to $ \tilde{C}$ such that $\tilde{\phi}(\mathbf{y})=\mathbf{y}\tilde{G}$ for any $\mathbf{y}\in \mathbb{F}_q^{k}$. Then $\tilde{\phi}= \varphi\phi$ by the definition of $\tilde{G}$. For any nonzero $\mathbf{c}\in C$, there is a $\mathbf{y}$ such $\mathbf{c}=\phi(\mathbf{y})$ and let $\tilde{\mathbf{c}}=\varphi(\mathbf{c})=\tilde{\phi}(\mathbf{y})$. Then  by Proposition~\ref{r pair weight}, we have
\begin{equation}\label{6.c}
  w_{p}(\mathbf{c})= n-\theta_{G}(\langle\mathbf{y}\rangle^{\bot})
\end{equation}
 and \begin{equation}\label{6.d}
  w_{p}(\tilde{\mathbf{c}})= n-\theta_{\tilde{G}}(\langle\mathbf{y}\rangle^{\bot}).
\end{equation}

For $0\leq r \leq s$, let $\Delta_{r}=(m_{G}(V^r_{1}),m_{G}(V^r_{2}),\cdots,m_{G}(V^r_{n_{r,k}}))$, and let $$\tilde{\Delta}_{r}=(m_{\tilde{G}}(V^r_{1}),m_{\tilde{G}}(V^r_{2}),\cdots,m_{\tilde{G}}(V^r_{n_{r,k}})).$$
Let $\Gamma_{k-1}=(\theta_{G}(V^{k-1}_{1}),\theta_{G}(V^{k-1}_{2}),\cdots,\theta_{G}(V^{k-1}_{n_{1,k-1}})),$  and let $$\tilde{\Gamma}_{k-1}=(\theta_{\tilde{G}}(V^{k-1}_{1}),\theta_{\tilde{G}}(V^{k-1}_{2}),\cdots,\theta_{\tilde{G}}(V^{k-1}_{n_{1,k-1}})).$$ Then we get
\begin{equation}\label{6.a}
  \Gamma_{k-1}=\sum_{r=0}^{s}\Delta_{r}T_{r,k-1}=m_G(\mathbf{0})\textbf{1}+\sum_{r=1}^{s}\Delta_{r}T_{r,k-1}
\end{equation}
and
\begin{equation}\label{6.b}
  \tilde{\Gamma}_{k-1}=\sum_{r=0}^{s}\tilde{\Delta}_{r}T_{r,k-1}=m_{\tilde{G}}(\mathbf{0})\textbf{1}+\sum_{r=1}^{s}\tilde{\Delta}_{r}T_{r,k-1}
\end{equation}
by the definition of $\theta_{G}$.

Suppose $a=w_b(\mathbf{c})-w_b(\tilde{\mathbf{c}})$ for any nonzero $\mathbf{c}\in C$. By Equation~\ref{6.c} and Equation~\ref{6.d}, we have $\theta_{G}(\langle\mathbf{y}\rangle^{\bot})-\theta_{\tilde{G}}(\langle\mathbf{y}\rangle^{\bot})=-a$ for any nonzero $\mathbf{y}\in \mathbb{F}_q^{k}$ and $$\Gamma_{k-1}-\tilde{\Gamma}_{k-1}=-a\mathbf{1}.$$

By Equation~\ref{6.a} and Equation~\ref{6.b}, we have $$\sum_{r=1}^{s}(\Delta_{r}-\tilde{\Delta}_{r})T_{r,k-1}=(m_{\tilde{G}}(\mathbf{0})-m_G(\mathbf{0})-a)\textbf{1}$$ and $$\sum_{r=1}^{s}(\Delta_{r}-\tilde{\Delta}_{r})T_{r,k-1}T_{1,k-1}^{-1}=\frac{m_{\tilde{G}}(\mathbf{0})-m_G(\mathbf{0})-a}{n_{1,k-1}}\mathbf{1}.$$
By Lemma~\ref{T1} (c), the element in the $i$th position of the vector $\sum_{r=1}^{s}(\Delta_{r}-\tilde{\Delta}_{r})T_{r,k-1}T_{1,k-1}^{-1}$ is
$$q\sum_{V\in \Omega_i  }\frac{1}{|V|}(m_G(V)-m_{\tilde{G}}(V))-\sum_{r=2}^s\sum_{V^r\in {\rm PG}^r(\mathbb{F}_q^k)}\frac{q^{r-1}-1}{q^{r-1}(q^{k-1}-1)}(m_G(V^r)-m_{\tilde{G}}(V^r)),$$
where $\Omega_i=\{V\in {\rm PG}^{\leq s}(\mathbb{F}_q^k)\,|\,V^1_i\subseteq V\}$. Then we have
\begin{eqnarray*}
&&q\sum_{V\in \Omega_i  }\frac{1}{|V|}(m_G(V)-m_{\tilde{G}}(V))-\sum_{r=2}^s\sum_{V^r\in {\rm PG}^r(\mathbb{F}_q^k)}\frac{q^{r-1}-1}{q^{r-1}(q^{k-1}-1)}(m_G(V^r)-m_{\tilde{G}}(V^r))\\
&=&\frac{m_{\tilde{G}}(\mathbf{0})-m_G(\mathbf{0})-a}{n_{1,k-1}}.
\end{eqnarray*}
Hence
\begin{eqnarray*}
q\sum_{V\in \Omega_i  }\frac{1}{|V|}(m_G(V)-m_{\tilde{G}}(V))&=&\sum_{r=2}^s\sum_{V^r\in {\rm PG}^r(\mathbb{F}_q^k)}\frac{q^{r-1}-1}{q^{r-1}(q^{k-1}-1)}(m_G(V^r)-m_{\tilde{G}}(V^r))\\
&+&\frac{m_{\tilde{G}}(\mathbf{0})-m_G(\mathbf{0})-a}{n_{1,k-1}}.
\end{eqnarray*}
This implies that  $\sum_{V\in \Omega_i  }\frac{1}{|V|}(m_G(V)-m_{\tilde{G}}(V))$ are constant vectors for any $1\leq i \leq n_{1,k}$.

Suppose $\sum_{V\in \Omega_i  }\frac{1}{|V|}(m_G(V)-m_{\tilde{G}}(V))=b$ for any $1\leq i \leq n_{1,k}$. Then $\sum_{r=1}^{s}(\Delta_{r}-\tilde{\Delta}_{r})T_{r,k-1}T_{1,k-1}^{-1}$ and $\sum_{r=1}^{s}(\Delta_{r}-\tilde{\Delta}_{r})T_{r,k-1}$ are constant vectors by Lemma~\ref{T1} (a),
since the element in the $i$th position of the vector $\sum_{r=1}^{s}(\Delta_{r}-\tilde{\Delta}_{r})T_{r,k-1}T_{1,k-1}^{-1}$ is
$$q\sum_{V\in \Omega_i  }\frac{1}{|V|}(m_G(V)-m_{\tilde{G}}(V))-\sum_{r=2}^s\sum_{V^r\in {\rm PG}^r(\mathbb{F}_q^k)}\frac{q^{r-1}-1}{q^{r-1}(q^{k-1}-1)}(m_G(V^r)-m_{\tilde{G}}(V^r))$$
$$=qb-\sum_{r=2}^s\sum_{V^r\in {\rm PG}^r(\mathbb{F}_q^k)}\frac{q^{r-1}-1}{q^{r-1}(q^{k-1}-1)}(m_G(V^r)-m_{\tilde{G}}(V^r)).$$
By Equation~\ref{6.a} and Equation~\ref{6.b}, we get $$\Gamma_{k-1}-\tilde{\Gamma}_{k-1}=\sum_{r=1}^{s}(\Delta_{r}-\tilde{\Delta}_{r})T_{r,k-1}+(m_{\tilde{G}}(\mathbf{0})-m_G(\mathbf{0}))\textbf{1}$$ is a constant vectors and $\theta_{G}(\langle\mathbf{y}\rangle^{\bot})-\theta_{\tilde{G}}(\langle\mathbf{y}\rangle^{\bot})$ is constant for any nonzero $\mathbf{y}\in \mathbb{F}_q^{k}$. Therefore, $w_b(\mathbf{c})-w_b(\varphi(\mathbf{c}))$ is constant for any nonzero $\mathbf{c}\in C$ by Equation~\ref{6.c} and Equation~\ref{6.d}.
\end{proof}

It is easy to get following result.
\begin{Corollary}\label{6.2}
Assume the notations are given above. Then $w_p(\mathbf{c})=w_p(\varphi(\mathbf{c}))$ for any $\mathbf{c}\in C$ if and only if $\sum_{V\in \Omega_i  }\frac{1}{|V|}(m_G(V)-m_{\tilde{G}}(V))$ is constant for any $1\leq i \leq n_{1,k}$ and there exists a nonzero $\mathbf{c_0}\in C$ such that $w_p(\mathbf{c_0})=w_p(\varphi(\mathbf{c_0}))$.

\end{Corollary}

From Theorems~\ref{2 condition} and \ref{6.1}, we know that if we want to  determine a linear code $C$ is or not a pair equiweight code and  a linear isomorphism is or not preserving pair weights of codes, it is crucial to calculate the following value $$\sum_{V\in \Omega_i  }\frac{1}{|V|}m_G(V)$$ for an $[n,k]$-linear code $C$ with a generator matrix $G$, where $s=\min\{2,k-1\}$ and $\Omega_i=\{V\in {\rm PG}^{\leq s}(\mathbb{F}_q^k)\,|\,V^1_i\subseteq V\}$.

Recall we assume that $G=(G_{0},\cdots,G_{n-1})$ is a generator matrix of an $[n,k]$-linear code $C$ over $\mathbb{F}_q$. Let $$S_j=\mathbb{F}_qG_{j}+\mathbb{F}_qG_{j+1}$$ which is an $\mathbb{F}_q$-subspace of $\mathbb{F}_q^k$ and let $$\tilde{S}_j=[ G_{j},G_{j+1}]$$ is an $k\times 2$ submatrix of $G$ for $0\leq j\leq n-1$. We know that $\dim(S_j)=rank(\tilde{S}_j)$, where $rank(\tilde{S}_j)$ denotes the rank of $\tilde{S}_j$.

The following theorem gives an algorithm to calculate the value $\sum_{V\in \Omega_i  }\frac{1}{|V|}m_G(V)$ in this section.

\begin{Theorem}\label{611}
Let  $\kappa_{ij}=\left\{ \begin{array}{ll}
1,  & \textrm{if $V^1_{i}\subseteq S_j ;$}\\
0,  & \textrm{if $V^1_{i}\nsubseteq S_j .$}
\end{array} \right.$ for $1\leq i\leq n_{1,k}$ and $1\leq j \leq n$. Then $$\sum\limits_{V\in \Omega_i  }\frac{1}{|V|}m_G(V)=\sum\limits_{j=1}^{n}\kappa_{ij}q^{-rank(\tilde{S}_j)}.$$
\end{Theorem}

\begin{proof}
It is easy to prove this theorem by using the definition of the function $m_{G}$.
\end{proof}

\begin{Remark}
Let $C$ be an $[n,k]$-linear code over $\mathbb{F}_q$ with a generator matrix $G=(G_{0},\cdots,G_{n-1})$, then we can calculate $f_i=\sum_{j=1}^{n}\kappa_{ij}q^{-rank(\tilde{S}_j)}$ for $1\leq i\leq n_{1,k}$. First we can calculate $\{S_0,S_1,\cdots,S_{n-1}\}$, and $|{\rm PG}^1(S_i)|\leq q+1$. Assume $T=\bigcup_{i=1}^{n}{\rm PG}^1(S_i)$, we have $|T|\leq n(q+1)$. If $V^1_i\notin T$, then $f_i=0$ by Theorem~\ref{611}. So we only need to calculate $|T|$ subspaces of one dimension of $\mathbb{F}_q^k$ for $f_i$.

However, if we simply look at all $q^k$ codewords of $C$ and check their pair weights, then we need to calculate $\frac{q^k-1}{q-1}$ subspaces of dimension one of $\mathbb{F}_q^k$ for their pair weights since $\mathbf{c}$ and $\lambda\mathbf{c}$ have same pair weight for $\mathbf{c}\in C$ and $\lambda\in \mathbb{F}_q^*$. So using our characterization to decide if $C$ is a pair equiweight code or if a linear isomorphism preserve pair weight is more efficiently, since $|T|\leq n(q+1)<<\frac{q^k-1}{q-1}$ when $q$ is large. For example, when $C$ is a $[10,5]$-linear code $C$ over $\mathbb{F}_{31}$, $|T|=320$ is much less than $\frac{31^5-1}{31-1}\approx 28629151$.

\end{Remark}
\begin{Example}
Let $C,C_1,C_2$ be  linear codes of length $4$ with generator matrices  $$G=\left(\begin{array}{cccc}
                       1 & 0&0&0 \\
                        0 &1&0&1 \\
                        0&0&1&0

\end{array}\right) ,$$
$$G_1=\left(\begin{array}{cccc}
                       0 & 0&1&0 \\
                        0 &1&0&1 \\
                        1&0&0&0

\end{array}\right) ,$$
$$G_2=\left(\begin{array}{cccc}
                       1 & 0&0&0 \\
                        0 &0&1&1 \\
                        0&1&0&0

\end{array}\right) $$
over $\mathbb{F}_2$, respectively. And let $\varphi_1:\,C\rightarrow C_1 $ and $\varphi_2:\,C\rightarrow C_2 $ be linear isomorphisms such that $$\varphi_1((c_0,c_1,c_2,c_3))=(c_2,c_1,c_0,c_3)$$ and $$\varphi_2((c_0,c_1,c_2,c_3))=(c_0,c_2,c_1,c_3)$$ for any $(c_0,c_1,c_2,c_3)\in C$. Assume
$$V_1^1=\left(\begin{array}{c}
                       1  \\
                        0  \\
                        0

\end{array}\right) ,V_2^1=\left(\begin{array}{c}
                       0  \\
                        1  \\
                        0

\end{array}\right) ,V_3^1=\left(\begin{array}{c}
                       0  \\
                        0  \\
                        1

\end{array}\right) ,V_4^1=\left(\begin{array}{c}
                       1  \\
                        1 \\
                        0

\end{array}\right)  ,$$
$$V_5^1=\left(\begin{array}{c}
                       0  \\
                        1  \\
                        1

\end{array}\right) , V_6^1=\left(\begin{array}{c}
                       1  \\
                        0  \\
                        1

\end{array}\right) , V_7^1=\left(\begin{array}{c}
                       1  \\
                        1  \\
                        1

\end{array}\right) .$$

By Theorem~\ref{611}, we get following sequences such that
$$\{\sum_{V\in \Omega_i  }\frac{1}{|V|}m_G(V), 1\leq i \leq 7\}=\{\frac{1}{2},1,\frac{1}{2},\frac{1}{2},\frac{1}{2},0,0\},$$
$$\{\sum_{V\in \Omega_i  }\frac{1}{|V|}m_{G_{1}}(V), 1\leq i \leq 7\}=\{\frac{1}{2},1,\frac{1}{2},\frac{1}{2},\frac{1}{2},0,0\}$$ and
$$\{\sum_{V\in \Omega_i  }\frac{1}{|V|}m_{G_{2}}(V), 1\leq i \leq 7\}=\{\frac{1}{2},1,\frac{1}{2},\frac{1}{4},\frac{1}{4},\frac{1}{4},0\}.$$

Hence $\varphi_1$ preserves the pair weight, but $\varphi_2$ does not preserves the pair weight by Corollary~\ref{6.2}. On the other hand, we can get same result by calculate directly.

\end{Example}

\vskip 4mm

\noindent {\bf Acknowledgement.} This work was supported by NSFC (Grant No. 11871025).

%\newpage


\begin{thebibliography}{99}

\bibitem{B}	P. Beelen, ``A note on the generalized Hamming weights of Reed-Muller codes," Applicable Algebra in Engineering, Communication and Computing, vol. 30, no. 3, pp. 233-242, 2019.

\bibitem{BGG} K. Bogart, D. Goldberg, and J. Gordon, ``An elementary proof of the MacWilliams theorem on equivalence of codes," Information and Control, vol. 37, no. 1, pp. 19-22, 1978.
\bibitem{CB} Y. Cassuto and M. Blaum, ``Codes for symbol-pair read channels, ``IEEE Transactions on Information Theory, vol. 57, no. 12, pp. 8011-8020, 2011.

\bibitem{C}	Y. M. Chee, L. Ji, H. M. Kiah, C. Wang, and J. Yin, ``Maximum distance separable codes for symbol-pair read channels," IEEE Transactions on Information Theory, vol. 59, no. 11, pp. 7259-7267, 2013.


\bibitem{CLL}B. Chen, L. Lin, and H. Liu, ``Constacyclic symbol-pair codes: lower bounds and optimal constructions," IEEE Transactions on Information Theory, vol. 63, no. 12, pp. 7661-7666, 2017.

\bibitem{DD}  G. B. M. V. Der Geer and M. V. Der Vlugt, ``On generalized Hamming weights of BCH codes," IEEE Transactions on Information Theory, vol. 40, no. 2, pp. 543-546, 1994.

\bibitem{DGZZ}	B. Ding, G. Ge, J. Zhang, T. Zhang, and Y. Zhang, ``New constructions of MDS symbol-pair codes," Designs, Codes and Cryptography, vol. 86, no. 4, pp. 841-859, 2018.

\bibitem{DNSS}	H. Q. Dinh, B. T. Nguyen, A. K. Singh, and S. Sriboonchitta, ``On the symbol-pair distance of repeated-root constacyclic codes of prime power lengths," IEEE Transactions on Information Theory, vol. 64, no. 4, pp. 2417-2430, 2017.

\bibitem{DWLS}	H. Q. Dinh, X. Wang, H. Liu, and S. Sriboonchitta, ``On the symbol-pair distances of repeated-root constacyclic codes of length $2p^s$," Discrete Mathematics, vol. 342, no. 11, pp. 3062-3078, 2019.

\bibitem{DHL} 	S. T. Dougherty, S. Han, and H. Liu, ``Higher weights for codes over rings," Applicable Algebra in Engineering, Communication and Computing, vol. 22, no. 2, pp. 113-135, 2011.



\bibitem{FL} Y. Fan and H. Liu, ``Generalized Hamming equiweight linear codes," Acta Electronica Sinica, vol. 31, no. 10, pp. 1591-1593, 2003.

\bibitem{F} J. G. D. Forney, ``Dimension/length profiles and trellis complexity of linear block codes," IEEE Transactions on Information Theory, vol. 40, no. 6, pp. 1741-1752, 1994.





\bibitem{JFW}G. Jian, R. Feng, and H. Wu, ``Generalized Hamming weights of three classes of linear codes," Finite Fields and Their Applications, vol. 45, no. 5,  pp. 341-354, 2017.


\bibitem{KZL}	X. Kai, S. Zhu, and P. Li, ``A construction of new MDS symbol-pair codes," IEEE Transactions on Information Theory, vol. 61, no. 11, pp. 5828-5834, 2015.

\bibitem{LG}	S. Li and G. Ge, ``Constructions of maximum distance separable symbol-pair codes using cyclic and constacyclic codes," Designs, Codes and Cryptography, vol. 84, no. 3, pp. 359-372, 2017.

\bibitem{LXY}	S. Liu, C. Xing, and C. Yuan, ``List decodability of symbol-pair codes," IEEE Transactions on Information Theory, vol. 65, no. 8, pp. 4815-4821, 2019.


\bibitem{M} J. MacWilliams, ``A theorem on the distribution of weights in a systematic code," Bell System Technical Journal, vol. 42, no. 1, pp. 79-94, 1963.


\bibitem{TV}  M. A. Tsfasman and S. G. Vladut, ``Geometric approach to higher weights," IEEE Transactions on Information Theory, vol. 41, no. 6, pp. 1564-1588, 1995.

%\bibitem{WW}  H. N. Ward and J. A. Wood, "Characters and the equivalence of codes," Journal of Combinatorial Theory, Series A, vol. 73, no. 2, pp. 348-352, 1996.

\bibitem{W} V. K. Wei, ``Generalized Hamming weights for linear codes," IEEE Transactions on information theory, vol. 37, no. 5, pp. 1412-1418, 1991.

\bibitem{W1}	E. Weiss, ``Linear codes of constant weight," SIAM Journal on Applied Mathematics, vol. 14, no. 1, pp. 106-111, 1966.

\bibitem{J}	J. A. Wood, ``Duality for modules over finite rings and applications to coding theory," American journal of Mathematics, vol. 121, no. 3, pp. 555-575, 1999.

\bibitem{W2} J. A. Wood, ``The structure of linear codes of constant weight," Transactions of the American Mathematical Society, vol. 354, no. 3, pp. 1007-1026, 2002.

\bibitem{J1}J. A. Wood, ``Code equivalence characterizes finite Frobenius rings," Proceedings of the American Mathematical Society, vol. 136, no. 2, pp. 699-706, 2008.






\bibitem{YBS}	E. Yaakobi, J. Bruck, and P. H. Siegel, ``Constructions and decoding of cyclic codes over $b$-symbol read channels," IEEE Transactions on Information Theory, vol. 62, no. 4, pp. 1541-1551, 2016.


\bibitem{YL}	M. Yang, J. Li, and K. Feng, ``Construction of cyclic and constacyclic codes for $b$-symbol read channels meeting the Plotkin-like bound," arXiv.org, 2016.


\bibitem{YF} 	M. Yang, J. Li, K. Feng, and D. Lin, ``Generalized Hamming weights of irreducible cyclic codes," IEEE Transactions on Information Theory, vol. 61, no. 9, pp. 4905-4913, 2015.











\end{thebibliography}
\end{document}